\documentclass{amsart}

\usepackage{amsmath,amsthm,amsfonts, amssymb,amscd}
\usepackage[all]{xy}

\newtheorem{theorem}{Theorem}[section]
\newtheorem{definition}[theorem]{Definition}%[section]
\newtheorem{example}[theorem]{Example}%[section]
\newtheorem{proposition}[theorem]{Proposition}

\newcommand{\minusre}{\hspace{0.3em}\raisebox{0.3ex}{\sl \tiny /}\hspace{0.3em}}
\newcommand{\minusli}{\hspace{0.3em}\raisebox{0.3ex}{\sl \tiny $\setminus $}\hspace{0.3em}}
\newcommand{\lex}{\,\overrightarrow{\times}\,}

\newcommand{\RDP}{\mbox{\rm RDP}}
\newcommand{\RIP}{\mbox{\rm RIP}}
\begin{document}
\title[The lexicographic product and pseudo effect algebras]{The Lexicographic Product of po-groups and $n$-perfect Pseudo Effect Algebras}
\author[A. Dvure\v censkij, J. Kr\v{n}\'avek]{Anatolij Dvure\v censkij$^{1,2}$ and Jan Kr\v{n}\'avek$^2$}
\date{}%Nov 1. 2010
\maketitle
\begin{center}  \footnote{Keywords: Pseudo effect algebra, po-group, $\ell$-group,  strong unit, Riesz Decomposition Property, lexicographic product, $n$-perfect pseudo effect algebra

 AMS classification: 81P15, 03G12, 03B50

The paper has been supported by Center of Excellence SAS -~Quantum
Technologies~-,  ERDF OP R\&D Project
meta-QUTE ITMS 26240120022, Slovak Research and Development Agency under the contract APVV-0178-11, the grant VEGA No. 2/0059/12 SAV and by
CZ.1.07/2.3.00/20.0051. }
Mathematical Institute,  Slovak Academy of Sciences,\\
\v Stef\'anikova 49, SK-814 73 Bratislava, Slovakia\\
$^2$ Depart. Algebra  Geom.,  Palack\'{y} Univer.\\
CZ-771 46 Olomouc, Czech Republic\\

E-mail: {\tt
dvurecen@mat.savba.sk} \qquad {\tt jan.krnavek@upol.cz}
\end{center}
%\normalsize

\begin{abstract} We will study the existence of different types of the Riesz Decomposition Property for the lexicographic product of two partially ordered groups. A  special attention will be paid to the lexicographic product of the group of the integers with an arbitrary po-group. Then we apply these results to the study of $n$-perfect pseudo effect algebras. We show that the category of strong $n$-perfect pseudo-effect algebras is categorically equivalent to the category
of torsion-free directed partially ordered groups with RDP$_1.$

\end{abstract}

\section{Introduction}%1

Immediately after appearing  quantum mechanics, it was observed that the classical laws of the Newton physics as well as of the probabilistic model of Kolmogorov fail in measurements of quantum observables.  Therefore, it was necessary to describe an appropriate algebraic model where we can model quantum phenomena. The first such model was a quantum logic which is an orthomodular poset. Latter there appeared orthomodular posets, orthoalgebras, and in 1994, Foulis and Bennett \cite{FoBe} suggested a successful model, effect algebras. It is a partial algebra with a partially defined addition, $+$, which means that $a+b$ denotes the conjunction of mutually excluding events $a$ and $b.$ Its main idea is to describe an appropriate model for the so-called POV-measures in the effect algebra $\mathcal E(H)$ of Hermitian operators between the zero and the identity operators of a real, complex or quaternionic Hilbert space $H.$

In many important cases, effect algebras are intervals in  partially ordered groups (= po-groups) with strong unit. Such an important result was established by Ravindran \cite{Rav} for  effect algebras with the Riesz Decomposition Property, RDP. We note that RDP means roughly speaking a weak form of the distributivity or a possibility of a joint refinement of two decompositions. It was observed also that the class of effect algebras contains the class of MV-algebras, \cite{Cha}, algebras which model many-valued reasoning. They can be characterized as lattice ordered effect algebras with RDP.

It is necessary to recall that RDP fails for $\mathcal E(H).$ However, using a delicate result of von Neumann \cite{vNe}, every set of mutually commuting elements in $\mathcal E(H)$ can be embedded into a maximal Abelian von Neumann algebra $\mathcal A$, and every element of $\mathcal A$   is a real-valued function of some element $a$ in $\mathcal A.$ This can be by \cite{Pul, Var} transformed into an MV-algebra (as functions of $a$) which is even  a $\sigma$-complete MV-algebra. In other words, $\mathcal E(H)$ can be covered by a system of $\sigma$-complete MV-algebras, and each MV-algebra satisfies RDP.  Therefore, $\mathcal E(H)$ satisfies RDP locally and not globally.

During the last period, effect algebras are intensively studied by many experts. As a guide trough the realm of effect algebras and of quantum structures we can recommend the monograph \cite{DvPu}.

Recently there appeared a noncommutative version of effect algebras, called pseudo effect algebras, \cite{DvVe1,DvVe2}, where it was not more assumed that $a+b=b+a$ holds for all $a,b.$ Also in this case it was shoved that in some important cases pseudo effect algebras are intervals in not necessarily Abelian po-groups with strong unit. This is true also for noncommutative versions of MV-algebras \cite{Dvu1}. Such a condition is a stronger version of RDP, called RDP$_1,$ which  coincides with RDP if the pseudo effect algebra is commutative = effect algebra, see \cite{DvVe1, DvVe2}, but for pseudo effect algebras they can be different.

A special class of effect algebras or pseudo effect algebras is formed by the so-called perfect algebras. That is a class of pseudo effect algebras $E$ such that every element $a\in E$ belongs either to the set of infinitesimals, that is, to the system of elements $x \in E$ such that the addition $nx= x+\cdots + x$ exists for any integer $n\ge 1,$ or it is a co-infinitesimal, i.e. it is negation of some infinitesimal.  Such effect algebras were studied in \cite{Dvu2} where it was shown that they are categorically equivalent with the category of po-groups with RDP. The class of perfect MV-algebras was described in \cite{DiLe} and the class of $n$-perfect noncommutative MV-algebras was studied in \cite{DDT}.

Recently, we have studied $n$-perfect pseudo effect algebras in \cite{DXY} for $n\ge 1.$ They allow us to decompose  pseudo effect algebras into $n+1$ comparable slices. They are non-Archimedean algebras and are connected with the so-called lexicographic product, $\mathbb Z \lex G,$ of the group of integers $\mathbb Z$ with some po-group $G.$ In that paper was not clear how is connected RDP$_1$ of $G$ with RDP$_1$ of $\mathbb Z \lex G,$ and with RDP$_1$ of  the corresponding interval pseudo effect algebra of $\mathbb Z\lex G,$ respectively. Therefore, the main aim of the present paper is to establish this connection in more details. This will allow us to simplify some important results from \cite{DXY}.

The paper is organized as follows. Section 2 presents the elements of theory of pseudo effect algebras and their connection to unital po-groups  with different types of RDP. In Section 3, we will study the lexicographic product of  po-groups with some kind of RDP. A special attention will be paid to the lexicographic product of the form $\mathbb Z \lex G.$ The final section contains an application of the properties of $\mathbb Z \lex G$ to the class of strong $n$-perfect pseudo effect algebras.

\section{Effect Algebras, Pseudo Effect Algebras, po-Groups, and the Riesz Decomposition Properties}%2

According to \cite{DvVe1, DvVe2}, a partial algebraic structure
$(E;+,0,1),$  where $+$ is a partial binary operation and 0 and 1 are constants, is called a {\it pseudo effect algebra},  if for all $a,b,c \in E,$ the following hold.
\begin{itemize}
\item[{\rm (PE1)}] $ a+ b$ and $(a+ b)+ c $ exist if and only if $b+ c$ and $a+( b+ c) $ exist, and in this case,
$(a+ b)+ c =a +( b+ c)$.

\item[{\rm (PE2)}] There are exactly one  $d\in E $ and exactly one $e\in E$ such
that $a+ d=e + a=1 $.

\item[{\rm (PE3)}] If $ a+ b$ exists, there are elements $d, e\in E$ such that
$a+ b=d+ a=b+ e$.

\item[{\rm (PE4)}] If $ a+ 1$ or $ 1+ a$ exists,  then $a=0$ .
\end{itemize}

If we define $a \le b$ if and only if there exists an element $c\in
E$ such that $a+c =b,$ then $\le$ is a partial ordering on $E$ such
that $0 \le a \le 1$ for any $a \in E.$ It is possible to show that
$a \le b$ if and only if $b = a+c = d+a$ for some $c,d \in E$. We
write $c = a \minusre b$ and $d = b \minusli a.$ Then

$$ (b \minusli a) + a = a + (a \minusre b) = b,
$$
and we write $a^- = 1 \minusli a$ and $a^\sim = a\minusre 1$ for any
$a \in E.$

A non-empty subset $I$ of a pseudo effect algebra $E$ is said to be an {\it ideal} if (i) $a,b\in I,$ $a+b\in E,$ then $a+b \in I,$ and (ii) if $a\le b \in I,$ then $a \in E.$ A mapping $h$ from a pseudo effect algebra $E$ into another one $F$ is said to be a {\it homomorphism} if (i) $h(1)=1,$ and (ii) if $a+b$ is defined in $E,$ so is defined $h(a)+h(b)$ and $h(a+b)= h(a)+h(b).$

If $A,B$ are two subsets of $E,$ $A+B:=\{a+b: a \in A, b \in B, a+b \in E\}.$

We recall that a {\it po-group} (= partially ordered group) is a
group $G$ endowed with a partial order, $\le,$ such that if $a\le b,$ $a,b
\in G,$ then $x+a+y \le x+b+y$ for all $x,y \in G.$  We denote by
$G^+$ the set of all positive elements of $G.$ If, in addition, $G$
is a lattice under $\le$, we call it an $\ell$-group (= lattice
ordered group). An element $u\in G^+$ is said to a {\it strong unit}
(= order unit) if given $g \in G,$ there is an integer $n\ge 1$ such
that $g \le nu,$ and the couple $(G,u)$ with a fixed strong unit is
called a {\it unital po-group.} The monographs \cite{Fuc, Gla} can
be recommended  as guides through the world of po-groups.

For example, if $(G,u)$ is a unital (not necessary Abelian) po-group
with strong unit $u$, and
$$
\Gamma(G,u) := \{g \in G: \ 0 \le g \le u\},\eqno(2.1)
$$
then $(\Gamma(G,u); +,0,u)$ is a pseudo effect algebra if we
restrict the group addition $+$ to $\Gamma(G,u).$  Every pseudo
effect algebra $E$ that is isomorphic to some $\Gamma(G,u)$ is said
to be an {\it interval pseudo effect algebra}.

For basic properties of pseudo effect algebras see \cite{DvVe1} and
\cite{DvVe2}. We recall that if $+$ is commutative, $E$ is said to
be an {\it effect algebra}; for a comprehensive overview on effect
algebras see e.g. \cite{DvPu}.

We note that if $a^-=a^\sim$ for all $a\in E,$ then it does not imply, in general, that $E$ is an effect algebra. Indeed, if $G$ is not an Abelian po-group, let $\mathbb Z\lex G$ denote the lexicographic product of the group of integer, $\mathbb Z$, with $G.$ Then $u=(1,0)$ is a strong unit, and $E=\Gamma(\mathbb Z \lex G,(1,0))$ is a pseudo effect algebra with $a^-=a^\sim$ for all $a\in E,$ but $E$ is not an effect algebra.

Let $a$ be an element of a pseudo effect algebra $E$ and $n\ge 0$ be an integer. We define
$$
0a:=0,\quad 1a:=a, \quad na:= (n-1)a +a, \ n \ge 2
$$
supposing $(n-1)a$ and $(n-1)a+a$ are defined in $E.$

A crucial effect algebra important for mathematical foundations of quantum physics is the system $\mathcal E(H)$ of all Hermitian operators of a real, complex or quaternionic Hilbert space $H$ which are between the zero $O$ and the identity operator $I$, \cite{DvPu}. It is also an interval effect algebra because if $\mathcal B(H)$ is the system of all Hermitian operators, then $I$ is a strong unit, and $\mathcal E(H)=\Gamma(\mathcal B(H),I).$

\begin{definition}\label{de:2.1}
A  po-group $G$ satisfies

\begin{enumerate}
\item[(i)]
the Riesz Interpolation Property (\RIP) if $a_1,a_2 \le b_1,b_2\in G$ implies there exists an element $c\in G$ such that $a_1,a_2 \le c \le b_1,b_2;$

\item[(i)]
the Riesz Decomposition Property (\RDP$_0$ for short) if for $a,b,c \in G^+,$ $a \le b+c$, there exist $b_1,c_1 \in G^+,$ such that $b_1\le b,$ $c_1 \le c$ and $a = b_1 +c_1;$

\item[(iii)]
the Riesz Decomposition Property (\RDP for short) if, for all $a_1,a_2,b_1,b_2 \in G^+$ such that $a_1 + a_2 = b_1+b_2,$ there are four elements $c_{11},c_{12},c_{21},c_{22}\in G^+$ such that $a_1 = c_{11}+c_{12},$ $a_2= c_{21}+c_{22},$ $b_1= c_{11} + c_{21}$ and $b_2= c_{12}+c_{22};$

\item[(iv)]
\RDP$_1$  if, for all $a_1,a_2,b_1,b_2 \in G^+$ such that $a_1 + a_2 = b_1+b_2,$ there are four elements $c_{11},c_{12},c_{21},c_{22}\in G^+$ such that $a_1 = c_{11}+c_{12},$ $a_2= c_{21}+c_{22},$ $b_1= c_{11} + c_{21}$ and $b_2= c_{12}+c_{22}$, and $0\le x\le c_{12}$ and $0\le y \le c_{21}$ imply  $x+y=y+x;$

\item[(v)]
\RDP$_2$  if, for all $a_1,a_2,b_1,b_2 \in G^+$ such that $a_1 + a_2 = b_1+b_2,$ there are four elements $c_{11},c_{12},c_{21},c_{22}\in G^+$ such that $a_1 = c_{11}+c_{12},$ $a_2= c_{21}+c_{22},$ $b_1= c_{11} + c_{21}$ and $b_2= c_{12}+c_{22}$, and $c_{12}\wedge c_{21}=0.$
\end{enumerate}
\end{definition}

If, for $a,b \in G^+,$ we have for all $0\le x \le a$ and $0\le b\le b,$ $x+y=y+x,$ we denote this property by $a\, \mbox{\rm \bf com}\, b.$

The RDP will be denote by the following table:

$$
\begin{array}{ccc}
a_1  & c_{11} & c_{12}\\
a_{2}  & c_{21} & c_{22}\\
             &b_{1} & b_{2}
\end{array}.
$$

By \cite[Prop 4.2]{DvVe1} for directed po-groups, we have
$$
\RDP_2 \quad \Rightarrow \RDP_1 \quad \Rightarrow \RDP \quad \Rightarrow \RDP_0 \quad \Leftrightarrow \quad  \RIP, \eqno(2.2)
$$
but the converse implications do not hold, in general. The equivalence of RDP$_0$ and RIP was established in \cite[Thm 2.2]{Fuc1}. A po-group $G$ satisfies \RDP$_2$ iff $G$ is an $\ell$-group, \cite[Prop 4.2(ii)]{DvVe1}.

If $G$ is an Abelian po-group, then
$$
\RDP_2 \quad \Rightarrow \RDP_1 \quad \Leftrightarrow \RDP \quad \Leftrightarrow \RDP_0 \quad \Leftrightarrow \quad  \RIP. \eqno(2.3)
$$

We define also the {\it Strict Riesz Interpolation Property} (SRIP for short) on a po-group $G$: given any $a_1,a_2,b_1,b_2\in G$ satisfying $a_1,a_2 < b_1,b_2,$ there exists $c \in G$ such that $a_1,a_2 < c < b_1,b_2.$

For example, the $\ell$-groups of real numbers $\mathbb R$ and rational numbers $\mathbb Q$ satisfy SRIP, but $\mathbb Z$ not.

In the same way as for po-groups, we can define analogous  Riesz Decomposition Properties, when we change a po-group $G$ and $G^+$ to a pseudo effect algebra $E.$

We see that RDP means roughly speaking a weak form of the distributivity and it allows a joint refinement of two decompositions.

The importance of RDP$_1$ follows from the crucial results by \cite{DvVe1, DvVe2} which says that if $E$ is a pseudo effect algebra with RDP$_1$, then there is a unique (up to isomorphism of unital po-groups) unital po-group $(G,u)$ with RDP$_1$ (not necessarily Abelian) such that $E \cong \Gamma(G,u).$ In addition, the category  of pseudo effect algebras with
RDP$_1$ and the category of unital po-groups with RDP$_1$
(not necessarily Abelian) are categorically equivalent. If $E$ is a pseudo effect algebra with RDP$_2,$ then $E$ is isomorphic to $\Gamma(G,u),$ where $G$ is an $\ell$-group, \cite{Dvu1, DvVe2}.

We say that (i) a po-group $G$ is {\it directed} if given $g_1,g_2 \in G$ there is $g \in G$ such that $g_1,g_2 \le g;$ (ii)  a group $G$ is {\it torsion-free} if $ng\ne 0$ for any $g\ne 0$ and every nonzero integer $n.$  For example, every $\ell$-group is torsion-free, see \cite[Cor 2.1.3]{Gla}. We recall that  a po-group $G$ is a torsion-free iff so is $\mathbb Z \lex G.$

\section{Lexicographic Product of po-groups and RDP's}%3

In this section, we study the lexicographic product of pseudo effect algebras in relationship  to the Riesz Decomposition Properties. We show in particular that if $G$ is a directed po-group with RIP and $E=\Gamma(\mathbb Z\lex G,(1,0))$ has  RDP or RDP$_1,$ then also the whole group $\mathbb Z\lex G$ has the same property RDP or RDP$_1.$

Let $\{G_i: i \in I\}$ be a system of po-groups indexed by a well-ordered index set $I.$ We define a {\it lexicographic order} $\le$ on the direct product $G=\prod_{i\in I} G_i$ which is given as follows:
for two different elements $a=(a_i)$ and $ b= (b_i),$  we have $a \le b$ iff $a_j < b_j,$ where $j$ is the least element of the set $\{i\in I: a_i\ne b_i\}.$  The po-group $G$ equipped with this lexicographic order is said to be the {\it lexicographic product} of the system $\{G_i: i \in I\}$, and we will denote it by $G = \lex_{i\in I}G_i.$

Let $G$ be a po-group. Since $\{0\} \lex G \cong G\cong G \lex \{0\},$ then $\{0\}\lex G$ and $G\lex \{0\}$ are directed iff $G$ is directed. If $G_1$ and $G_2$ are nontrivial po-groups, then $G_1 \lex G_2$ is directed iff so is $G_1$, \cite[p. 26 (c)]{Fuc}.

\begin{proposition}\label{pr:2.2} Let $\{G_i: i\in I\}$ be a family of po-groups and let $G =\lex_{i\in I} G_i.$

{\rm (i)} Let every $G_i$ satisfy \RIP and assume that every $G_i$ satisfies also {\rm SRIP} or that $G_j$ is directed for all $j>i.$ Then $G$ satisfies \RIP.

{\rm (ii)} $G$ satisfies \RDP$_2$ if and only if all $G_i$ are linearly ordered except the last one, if such exists, which may be an arbitrary po-group satisfying \RDP$_2.$
\end{proposition}

\begin{proof}
(i)  The proof of this part is identical with the proof of \cite[Prop 2.10]{Goo}.

(ii) According to \cite[Prop 4.2(ii)]{DvVe1}, a po-group $G$ satisfies \RDP$_2$ iff $G$ is an $\ell$-group. Due to the characterization \cite[p. 26 (d)]{Fuc}, we obtain the statement in question.
\end{proof}

\begin{proposition}\label{pr:2.3} Let $G_1,G_2$ be po-groups and $G = G_1 \lex G_2.$

{\rm (1)} $G$ satisfies  \RIP  if and only if the following properties hold:
\begin{enumerate}

\item[{\rm (a)}] Both $G_1$ and $G_2$ satisfy the RIP

\item[{\rm (b)}] Either $G_1$ satisfies {\rm SRIP} or $G_2$ is directed.

\end{enumerate}

{\rm (2)} If $G$ satisfies one of the properties \RDP or \RDP$_1,$ then  the following properties hold:
\begin{enumerate}

\item[{\rm (a)}] Both $G_1$ and $G_2$ satisfy the \RDP or \RDP$_1$.

\item[{\rm (b)}] Either $G_1$ satisfies {\rm SRIP} or $G_2$ is directed.
\end{enumerate}

{\rm (3)} $G$ satisfies \RDP$_2$ if and only if $G_1$ is linearly ordered and $G_2$ satisfies \RDP$_2$.
\end{proposition}

\begin{proof}
(1) \RIP: If (a) and (b) holds, then $G$ has RIP by Proposition \ref{pr:2.2}. Conversely, let $G$  satisfies \RIP. Since $\{0\}\times G_2$ is a convex subgroup of $G,$ it also satisfies \RIP. But $\{0\}\times G_2$ is isomorphic as an ordered po-group with $ G_2,$ so that $G_2$ satisfies RIP.

Let now $a_1,a_2, b_1, b_2\in G_1$ satisfy $a_i \le b_j$ for all $i,j.$ Then $(a_1,0),(a_2,0)\le (b_1,0),(b_2,0).$ The RIP holding for $G$ entails that there is an element $(c,d) \in G$ such that $(a_i,0)\le (c,d)\le (b_j,0)$ for all $i,j.$ Hence, $a_i\le c\le b_j$ for all $i,j$ which yields $G_1$ satisfies \RIP.

Now we show that if $G_2$ is not directed, then $G_1$ satisfies SRIP.  Thus let $u_1$ and $u_2$ be two elements which have no common upper bound in $G_2.$ Let $a_1,a_2,b_1,b_2\in G_1$ satisfy $a_1,a_2 < b_1,b_2.$ Then $(a_i,u_k)<(b_j,-u_k)$ for all $i,j,k.$ Therefore, there is an element $(z,w)\in G$ such that
$$
(a_1,u_k),(a_2,u_k)\le (z,w)\le (b_1,-u_k),(b_2,-u_k)
$$
for each $k.$ If $a_1=z,$ then $(a_1,u_k)\le (a_1,w)$ for each $k$ and so each $u_k \le w$ which is impossible. Consequently, $a_1 <z$ and similarly, $a_2 <z$. If $z=b_1$, then $(b_1,w)\le(b_1,-u_k)$ for each $k$ which entails $w \le -u_k$ and $u_k \le -w$ which again gives a contradiction. Thus $a_1,a_2 <x<b_1,b_2$ proving that $G_1$ satisfies SRIP.

(2) \RDP: Assume that $G$ satisfies RDP. Similarly as in the case RIP, $\{0\}\times G_2$ is a convex subgroup of $G$ satisfying RDP. Since $\{0\}\times G_2 \cong G_2$ as po-groups, $G_2$ satisfies RDP.

Given $a_1,a_2,b_1, b_2 \in G_1^+$ satisfying $a_1+a_2=b_1 + b_2,$ we have $ (a_1,0)+ (a_2,0)=(b_1,0) + (b_2,0).$ The RDP holding in $G$ implies that there are four elements $(c_{11},x_{11}),$ $(c_{12},x_{12}),(c_{21},x_{21}), (c_{22},x_{22})\in G^+$ such that $(a_1,0) = (c_{11},x_{11})+(c_{12},x_{22}),$ $(a_2,0)= (c_{21},x_{21})+(c_{22},x_{22}),$ $(b_1,0)= (c_{11},x_{11}) + (c_{21},x_{21})$ and $(b_2,0)= (c_{12},x_{12})+(c_{22},x_{22}).$ We have that all $c_{ij}$'s belong to $G_1^+.$ Therefore,  $a_1 = c_{11}+c_{12},$ $a_2= c_{21}+c_{22},$ $b_1= c_{11} + c_{21}$ and $b_2= c_{12}+c_{22},$ which yields that $G_1$ satisfies RDP.

RDP$_1$: Assume that $G$ satisfies RDP$_1$. Similarly as in the case for RDP, we can show that $\{0\}\times G_2$ satisfies RDP$_1,$ and since $\{0\} \times G_2 \cong G_2,$ $G_2$ is with RDP$_1.$

Given $a_1,a_2,b_1, b_2 \in G_1$ satisfying $a_1+a_2=b_1 + b_2,$ we have $ (a_1,0)+ (a_2,0)=(b_1,0) + (b_2,0).$ As for RDP, we  found four elements
$(c_{11},x_{11}),(c_{12},x_{12}),(c_{21},x_{21}),$ $(c_{22},x_{22})\in G^+$ such that $(a_1,0) = (c_{11},x_{11})+(c_{12},x_{22}),$ $(a_2,0)= (c_{21},x_{21})+(c_{22},x_{22}),$ $(b_1,0)= (c_{11},x_{11}) + (c_{21},x_{21})$ and $(b_2,0)= (c_{12},x_{12})+(c_{22},x_{22}).$ We have that all $c_{ij}$'s belong to $G_1^+.$ Therefore,  $a_1 = c_{11}+c_{12},$ $a_2= c_{21}+c_{22},$ $b_1= c_{11} + c_{21}$ and $b_2= c_{12}+c_{22}.$ In addition, $x_{11}=-x_{12}= -x_{21}=x_{22}= d$ for some element $d \in G_2.$
Assume that for $x,y \in G^+_1,$ we have $x\le c_{12}$ and $x\le c_{21}.$ Then $(x,-d)\le (c_{12},-d)$ and $(y,-d)\le (c_{21},-d)$ which yields $(x,-d)+(y,-d)=(y,-d)+(x,-d)$ and $x+y=y+x$ proving $G_1$ satisfies RDP$_1.$

(3) It follows from \cite[(d) page 26]{Fuc}.
\end{proof}

We note that according to \cite[Cor 2.12]{Goo}, for Abelian po-groups, the conditions (a)--(b) are also sufficient in order to  $G=G_1 \lex G_2$ satisfy RDP while in this case RIP = RDP = RDP$_1$. We do not know the answer for po-groups which are not Abelian. In the rest  of this section, we exhibit the case when $G_1 = \mathbb Z,$ and show some positive  answers, Theorems \ref{th:2.8} and \ref{th:2.11}. On the other hand,  according to \cite[Ex 2.13]{Goo}, there is an Abelian po-group $G$ with RIP such that the po-group $G\lex G$ does not satisfy RIP.

\begin{proposition}\label{pr:3.1}
Let $G$ be a directed po-group satisfying \RDP or {\rm RDP}$_1,$ respectively. Then the pseudo effect algebra $\Gamma(\mathbb Z \lex G,(1,0))$ satisfies \RDP or {\rm RDP}$_1,$ respectively.
\end{proposition}

\begin{proof}
Define the pseudo effect algebra $E = \Gamma(\mathbb Z \lex G, (1,0)).$ Then $a \in E$ iff $a= (0,g)$ or $a=(1,-g),$ where $g \in G^+$. If $(0,g_1)+(0,g_2)=(0,h_1)+ (0,h_2)$ for some $g_1,g_2,h_1,h_2 \in G^+,$ then RDP or RDP$_1$ holds.

(i) Assume that the case $(1,a_{1}) + (0,a_{2})=(1,b_{1})+(0,b_{2})$. Then $a_{2},b_{2} \in G^+.$

Let $(1,a_1)+(0,a_2)=(1,b_1)+(0,b_2).$ Then $a_2,b_2 \ge 0$ and $a_1,b_1\le 0.$  Let $d \le a_1,b_1.$  Then there exist $e_{11},e_{12},e_{21},e_{22} \in G^+$ such that

$$
\begin{array}{ccc}
-d+ a_{1}  &  e_{11} & e_{12}\\
a_{2}   & e_{21} & e_{22} \\
           &-d+ b_{1} & b_{2}
\end{array},
$$
which yields

$$
\begin{array}{ccc}
a_{1}  & d+ e_{11} & e_{12}\\
a_{2}   & e_{21} & e_{22} \\
           &b_{1} & b_{2}
\end{array},
$$
and

$$
\begin{array}{ccc}
(1,a_{1})  & (1,d+ e_{11}) & (0,e_{12})\\
(0,a_{2})   & (0,e_{21}) & (0,e_{22}) \\
           &(1,b_{1}) & (0,b_{2})
\end{array}.
$$
This entails RDP as well as RDP$_1$ for this case if $G$ does satisfy RDP or RDP$_1,$ respectively.

In the same way we prove the case $(0,a_1)+(1,a_2)=(0,b_1)+(1,b_2).$

(ii) The case $(1,a_1)+(0,a_2)=(0,b_1)+(1,b_2).$  Then $a_1,b_2\le 0$ and $a_2,b_1 \ge 0$ and $-(0,b_1)+ (1,a_1)= (1,b_2)-(0,a_2)$ which entails $-b_1 + a_1 = b_2 - a_2\le 0$ and
the RDP decomposition in question

$$
\begin{array}{ccc}
(1,a_{1})  & (0,b_1) & (1,-b_1+a_1)\\
(0,a_{2})   & (0,0) & (0,a_2) \\
           &(0,b_{1}) & (1,b_{2})
\end{array}.
$$
Hence RDP holds. If now $(0,0)\le (x_1,x_2)\le (0,0)$ and $(0,0)\le (y_1,y_2)\le (1,-b_1+a_1),$ then $(x_1,x_2)=(0,0)$ so that trivially $(x_1,x_2)+(y_1,y_2)=(y_1,y_2) + (x_1,x_2)$ and RDP$_1$ holds  for this case, too.

Finally, the last case $(0,a_1)+(1,a_2)=(1,b_1)+(0,b_2)$ follows from the equality $(1,b_1)+(0,b_2)=(0,a_1)+(1,a_2)$ and the latter case.

Combining all cases, we see that the pseudo effect algebra $E=\Gamma(\mathbb Z \lex G,(1,0))$ satisfies RDP or RDP$_1,$ respectively, whenever the po-group $G$ does  RDP or RDP$_1,$ respectively.
\end{proof}

Let $m,n\ge1$ be two integers. We say that a pseudo effect algebra $E$ satisfies the $(m,n)$-RDP  if $a_1+\cdots +a_m = b_1+\cdots+b_n$ implies that there is a system of elements $\{c_{ij}: 1\le i\le m,\ 1\le j \le n\}$ such that $a_i=\sum_{j=1}^n c_{ij}$ and $b_j = \sum_{i=1}^m c_{ij}$ for all $1\le i\le m,\ 1\le j \le n.$ If the elements $\{c_{ij}\}$ can be chosen in such a way that
and such that for $1 \leq i < m$, $\,1 \leq j < n$ we have

$$
c_{i+1,j} + \cdots + c_{mj} \,\mbox{\rm \bf com}\, c_{i,j+1} + \cdots + c_{in}, \eqno(3.1)
$$
we say that $E$ satisfies the $(m,n)$-RDP$_1.$
In the same way we introduce the $(m,n)$-RDP and the $(m,n)$-RDP$_1$ for a po-group $G$ assuming $a_1,\ldots,a_m,b_1,\ldots,b_n, c_{ij}\in G^+.$

If for two particular finite sequences $\{a_i\}_{i=1}^m$ and $\{b_j\}_{j=1}^n$ in $E$ or $G^+$ with $a_1+\cdots+a_m=b_1+\cdots + b_n$ we can find a finite system $\{c_{ij}\}$ in $E$ or $G^+,$ respectively, such that $a_i=\sum_{j=1}^n c_{ij}$ and $b_j = \sum_{i=1}^m c_{ij}$ for all $1\le i\le m,\ 1\le j \le n,$ we say that for $\{a_i\}$ and $\{b_j\}$ the RDP, $(m,n)$-RDP or RDP$_1,$ $(m,n)$-RDP$_1$ (if also (3.1) is true) holds.

In \cite[Lem 3.9]{DvVe1}, it was shown that RDP$_1$ holds in a pseudo effect algebra $E$ (in a po-group $G)$ iff $(m,n)$-RDP$_1$ holds for all integers $m,n\ge 1.$ In the same way, see also the proof of Proposition \ref{pr:2.6}, we can prove:

\begin{proposition}\label{pr:2.5}
The $\RDP$ holds in a pseudo effect algebra $E$ (in a po-group $G$) if and only if $(m,n)$-$\RDP$ holds for all integers $m,n\ge 1.$
\end{proposition}

\begin{proposition}\label{pr:2.6}
Let $(G,u)$ be a unital po-group satisfying \RIP such that $E=\Gamma(G,u)$ is a pseudo effect algebra satisfying \RDP or \RDP$_1.$
%{\rm (1)} If $a_1+a_2 = b_1+b_2,$ where $a_1,b_1\in E$ and $a_2,b_2 \in G^+,$ then for $\{a_1,a_2\}$ and $\{b_1,b_2\}$ \RDP or \RDP$_1$ holds whenever \RDP or \RDP$_1,$ respectively,  does hold for $G.$
%
%{\rm (2)}
If $a_1+\cdots+a_m = b_1+\cdots +b_n,$ where $a_1,\ldots,a_{m-1},b_1,\ldots,b_{n-1} \in E,$ and $a_m,b_n \in G^+,$ then for $\{a_i\}$ and $\{b_j\}$ $(m,n)$-\RDP or $(m,n)$-\RDP$_1$ holds whenever $(m,n)$-\RDP or $(m,n)$-\RDP$_1,$ respectively,  does hold for $G.$

\end{proposition}

\begin{proof}
(1)  The assertion is trivially satisfied for $n=1.$
Assume that $a_1+a_2 = b_1+b_2,$ where $a_1,b_1\in E$ and $a_2,b_2 \in G^+.$ Then $a_1,b_1
\le a_1+a_2 = b_1+b_2, u$ and the RIP entails that there is an element $z \in E$ such that $a_1,b_1\le z \le a_1+a_2,u.$  There are two elements $x,y \in E$ such that $a_1+x=z= b_1+y$ and the RDP yields four elements $c_{ij}$ such that

$$
\begin{array}{ccc}
a_1 & c_{11} & c_{12}\\
x   & c_{21} & c_{22} \\
           &b_1 & y
\end{array}.
$$

We can find an element $v \in G^+$ such that $a_1+a_2= z+v = b_1 + b_2.$ Then $ a_1+a_2 = z+v = a_1+x + v,$ i.e. $a_2= x+v,$ and $b_1+b_2=z+v= b_1+y+v$ so that $b_2=y+v.$ Hence, we have

$$
\begin{array}{ccc}
a_1 & c_{11} & c_{12}\\
a_2   & c_{21} & c_{22} +v \\
           &b_1 & b_2
\end{array},
$$
which implies RDP or RDP$_1$ for our elements $a_1,a_2,b_1,b_2.$

(2)  Assume $a_1+\cdots+a_m = b_1+\cdots +b_n,$ where where $a_1,\ldots,a_{m-1},b_1,\ldots,b_{n-1} \in E,$ and $a_m,b_m \in G^+.$ We asserts that for this case RDP or RDP$_1$ holds whenever they hold in $E.$
By RDP$_1$ and (1), the statement is true for $m,n\le 2.$ Suppose that the statement is true for all pairs of integers $m',n'$ such that $m'\le m,$ $m\ge 2$ is fixed,  and $n'< n,$ where $n\ge 3.$ Then $a_1+\cdots+a_m = b_1+\cdots + b_{n-1} + b_n +b_{n+1} =
b_1+\cdots +b_{n-1} + (b_n +b_{n+1}).$ By the induction hypothesis, there are elements $c_{ij},d_1,\ldots,d_{m-1}\in E$ and $d_m \in G^+$ such that

$$
\begin{array}{ccccc}
a_1 & c_{11} & \cdots & c_{1,n-1}&d_1\\
\vdots &  & & & \\
a_m   & c_{m1} & \cdots & c_{m,n-1} & d_{m} \\
           &b_1 & \cdots& b_{n-1} & b_n + b_{n+1}
\end{array}.
$$

Then for $d_1+\cdots +d_m = b_{n} + b_{n+1}$ we can apply the induction hypothesis, so that there are additional elements $c_{ij} \in G^+$ such that

$$
\begin{array}{ccc}
d_1 & c_{1n} & c_{1,n+1}\\
\vdots &  \\
d_m   & c_{mn} &  c_{m,n+1}  \\
       &    b_n & b_{n+1}
\end{array}.
$$
Putting together, we have

$$
\begin{array}{cccccc}
a_1 & c_{11} & \cdots & c_{1,n-1}&c_{1n}&c_{1,n+1}\\
\vdots &  & & & \\
a_m   & c_{m1} & \cdots & c_{m,n-1} & c_{mn}&c_{m,n+1}\\
           &b_1 & \cdots& b_{n-1} & b_n & b_{n+1}
\end{array},
$$
which entails $(m,n)$-RDP for this case. To prove $(m,n)$-RDP$_1,$ we have by the induction
$$
c_{i+1,j} + \cdots + c_{mj} \,\mbox{\rm \bf com}\, c_{i,j+1} + \cdots + c_{i,n-1} + d_i.
$$
Now take $x,y \in G^+$ such that $ x \le c_{i+1,j} + \cdots + c_{mj}$ and
$y \le c_{i,j+1} + \cdots + c_{i,n-1} + c_{i,n} +c_{i,n+1}.$ By RDP$_0,$ we have $y=y_1+\cdots +y_{n-1}+y_n +y_{n+1},$ where $y_k \le c_{i,k+1}$ for $k=j+1,\ldots, n+1.$ Since $y_n +y_{n+1}\le d_i,$ we see that $(m,n)$-RDP$_1$ holds for our case.
\end{proof}

\begin{theorem}\label{th:2.7}
Let $(G,u)$ be a unital po-group such that $\Gamma(G,u)$ satisfies \RIP. If the pseudo effect algebra $E=\Gamma(G,u)$ satisfies \RDP$_1,$ then  $(G,u)$ does $\RDP_1.$
\end{theorem}

\begin{proof}
Let $a_1+a_2 = b_1 + b_2,$ where $a_1,a_2,b_1,b_2 \in G^+.$ By Proposition \ref{pr:2.6}, the statement is true for $a_1,b_1 \le u$ and all $a_2,b_2 \in G^+.$ Assume by the induction hypothesis that the theorem is true for any $a_1 \le mu$ for a fixed integer $m\ge 1,$ for any $b_1 \le nu,$ $n\ge 1,$ and all $a_2,b_2 \in G^+.$ Thus let $b_1 = b_{11}+\cdots+b_{in}+b_{1,n+1},$ where all $b_j \le u.$ We prove the statements for $n+1.$ Then $a_1 + a_2 = b_1'+b_2 ', $ where $b_1'= b_{11}+\cdots+ b_{in}$ and $b_2'= b_{1,n+1}+b_2.$ By the induction hypothesis, there are elements $\{c_{ij}\}$ in $G^+$ such that

$$
\begin{array}{ccc}
a_1 & c_{11} & c_{12}\\
a_2   & c_{21} & c_{22} \\
           &b_1' & b_2'
\end{array},
$$
and $c_{12}  \,\mbox{\rm \bf com}\, c_{21}.$ Similarly there are elements $\{d_{ij}\}$ in $G^+$ such that

 $$
\begin{array}{ccc}
c_{12} & d_{11} & d_{12}\\
c_{22}   & d_{21} & d_{22} \\
           &b_{1,n+1} & b_2
\end{array},
$$
and $d_{12} \,\mbox{\rm \bf com}\, d_{21}.$  Whence

$$
\begin{array}{cccc}
a_1 & c_{11} & d_{11}& d_{12}\\
a_2   & c_{21} & d_{21}& d_{22} \\
           &b_1' & b_{1,n+1}& b_2
\end{array}
$$

and
$$
\begin{array}{ccc}
a_1 & c_{11}+d_{11} & d_{12}\\
a_2   & c_{21}+d_{21} & d_{22} \\
           &b_1 & b_2
\end{array}.
$$
The last table holds thanks to RDP$_1$. By this property, we have also
$c_{21} + d_{12} \,\mbox{\rm \bf com}\, d_{21}.$
\end{proof}

Thanks to the last theorem, we can extend Proposition \ref{pr:3.1}.

\begin{theorem}\label{th:2.8}
A po-group $\mathbb Z \lex G$ satisfies \RDP$_1$ if and only if $G$ is a directed po-group with $\RDP_1.$
\end{theorem}

\begin{proof}
If $\mathbb Z \lex G$ satisfies RDP$_1$, by Proposition \ref{pr:2.3}(2), the po-group $G$ is directed and with RDP$_1.$

Conversely, let $G$ be directed and with RDP$_1.$
By Proposition \ref{pr:2.3}(1), the po-group $\mathbb Z \lex G$ satisfies RIP and by Proposition \ref{pr:3.1}, the pseudo effect algebra $E=\Gamma(\mathbb Z\lex G,(1,0))$ satisfies RDP$_1.$ Applying Theorem \ref{th:2.7}, we see that $\mathbb Z \lex G$ satisfies RDP$_1.$
\end{proof}

Let $E$ be a pseudo effect algebra. A pair $((G,u),\iota_E),$ where $(G,u)$ is a unital po-group and $\iota_E$ is a homomorphism from $E$ into the pseudo effect algebra $\Gamma(G,u)$, is said to be a {\it universal group} for $E$ if, for every unital po-group $(H,v)$ and  for every homomorphism  $h: E \to \Gamma(H,v),$ there is a po-homomorphism of unital po-groups $k:(G,v)\to (H,v)$ such that $h=k\circ \iota_E.$ Then $h$ is a unique po-homomorphism satisfying $h=k\circ \iota_E.$

By \cite[Thm 7.2]{DvVe2} and by the proof of \cite[Thm 7.4]{DvVe2}, if $(G,u)$ is a unital po-group satisfying RDP$_1,$ then $(G,u,\iota),$ is a universal po-group for the pseudo effect algebra $E=\Gamma(G,u),$ where $\iota$ is the natural embedding from $E$ into $G.$

The following two examples show that Theorem \ref{th:2.7} cannot be extended to the case of all unital po-groups $(G,u)$ such that $E=\Gamma(G,u)$ satisfies RDP$_1$ but RIP fails for $G.$

\begin{example}\label{ex:2.9}{\rm We endow the group $G:=\mathbb Z \times Z$ with the partial ordering $\le$: $(a,b)\le (c,d)$ iff $a\le c$ and $b\le d$ or $(c+d)-(a+b) \ge 2.$ Then the element $u=(1,0)$ is a strong unit for $G,$ and $E:=\Gamma(G,u)=\{(0,0),(1,0)\}.$ The RIP fails for $G$ because for the elements $(1,0),(0,1) \le (0,3),(3,0)$ there is no element $(a,b)$ such that $(1,0),(0,1) \le (a,b)\le (0,3),(3,0).$ On the other hand, $E$ satisfies RDP, RDP$_1,$ and RDP$_2,$ but for $G,$ RDP, RDP$_1,$ and RDP$_2,$ fail.

We note that the universal group for $E$ is isomorphic with $(\mathbb Z,1)$ with the natural embedding $(0,0)\mapsto 0$ and $(1,0)\mapsto 1.$

}
\end{example}

\begin{example}\label{ex:2.10}{\rm Let again $G = \mathbb Z \times Z$ and we define a new partial order $\le$:  $(a,b)\le (c,d)$ if $a=c$ and $b=d$ or $a+b <c +d.$ Then $u=(1,0)$ is a strong unit, $E=\Gamma(G,u)=\{(0,0), (1,0)\},$ satisfies RDP, RDP$_1$ and RDP$_2$ but RIP  fails for $G$~-  take $(1,0),(0,1)\le (0,2),(2,0).$

}
\end{example}

In Proposition \ref{pr:3.1}, we have proved that if $G$ is a directed po-group with RDP, then $E=\Gamma(\mathbb Z \lex G, (1,0))$ satisfies RDP, too. But we do not know whether $\mathbb Z \lex G$ does satisfy RDP. The following theorem gives a positive solution.

\begin{theorem}\label{th:2.11} A po-group $\mathbb Z \lex G$ satisfies \RDP if and only if $G$ is a directed po-group with $\RDP.$
\end{theorem}

\begin{proof}
If $\mathbb Z \lex G$ satisfies RDP, by Proposition \ref{pr:2.3}(2), the po-group $G$ is directed and with RDP.

Conversely, let $G$ be directed and with RDP.
The elements of $(\mathbb Z \lex G)^+$ are of the form
$\{(0,g): g \in G^+\} \cup \{(n,g): n \ge 1,\ g \in G\}.$

(i) Let $(0,a_1)+(0,a_2) = (0,b_1)+ (0,b_2).$  Since then $a_1,a_2,b_1,b_2 \in G^+,$ the RDP follows from RDP for $G.$

(ii) Let $\min\{m_1,m_2,n_1,n_2\}=0$ and $m_1+m_2 >0.$ If we use the same ideas as  case (i) of Proposition \ref{pr:3.1}, we have the following decomposition tables

$$
\begin{array}{ccc}
(m_1,a_{1})  &  (n_1,b_1) & (n_2,-b_1+a_1)\\
(0,a_{2})   & (0,0) & (0,a_2) \\
           &(n_1,b_{1}) & (n_2,b_{2})
\end{array}, \ \mbox{for}\ n_2 >0,
$$

$$
\begin{array}{ccc}
(m_1,a_{1})  &  (n_1,d+e_{11}) & (0,e_{12})\\
(0,a_{2})   & (0,e_{21}) & (0,e_{22}) \\
           &(n_1,b_{1}) & (0,b_{2})
\end{array}, \ \mbox{for}\ n_2 =0,
$$
where $d\le a_1,b_1.$ The case $(m_1,a_1)+(0,a_2)=(0,b_1)+(n_2,b_2)$ follows from the first case of (ii).
%$$
%\begin{array}{ccc}
%(m_1,a_{1})  &  (0,b_1) & (n_2,-b_1+a_1)\\
%(0,a_{2})   & (0,0) & (0,a_2) \\
%           &(0,b_{1}) & (n_2,b_{2})
%\end{array}, \ \mbox{for}\ n_1 = 0.
%$$

In a similar way we deal with the case $(0,a_1)+(m_2,a_2) = (n_1,b_1)+ (n_2,b_2),$ where  $n_1 \ge 1.$ Then $m_2 \ge 1,$  $a_1 \ge 0,$ and we use the decomposition

$$
\begin{array}{ccc}
(0,a_{1})  &  (0,a_1) & (0,0)\\
(m_2,a_{2})   & (n_1,-a_1+b_1) & (n_2,b_2) \\
           &(n_1,b_{1}) & (n_2,b_{2})
\end{array}.
$$

(iii) Let $\min\{m_1,m_2,n_1,n_2\}\ge 1.$ We use the tables

$$
\begin{array}{ccc}
-d+a_{1}  &  e_{11} & e_{12}\\
a_{2}-d   & e_{21} & e_{22} \\
           &-d+ b_{1} & b_{2}-d
\end{array},
$$
where $d \le a_1,a_2,b_1,b_2,$ and then
$$
\begin{array}{ccc}
(m_1,a_{1})  &  (n_1,d+e_{11}) & (m_1-n_1,e_{12})\\
(m_2,a_{2})   & (0,e_{21}) & (m_2,e_{22}+d) \\
           &(n_1,b_{1}) & (n_2,b_{2})
\end{array} \ \mbox{for}\  m_1 \ge n_1,
$$
where $d \le a_1,a_2,b_1,b_2,$ and finally
$$
\begin{array}{ccc}
(m_1,a_{1})  &  (m_1,d+e_{11}) & (0,e_{12})\\
(m_2,a_{2})   & (n_1-m_1,e_{21}) & (n_2,e_{22}+d) \\
           &(n_1,b_{1}) & (n_2,b_{2})
\end{array} \mbox{for}\  n_1 \ge m_1,
$$
where $d \le a_1,a_2,b_1, b_2.$
\end{proof}

\section{Application to $n$-perfect Pseudo Effect Algebras}%4

We will study a special class of pseudo effect algebras which can be decomposed into $n+1$ comparable slices; we call them $n$-perfect. A prototypical example is $E=\Gamma(\mathbb Z\lex G,(n,0)).$ Such pseudo effect algebras are perfect when $n=1,$ in general. Perfect effect algebras were studied in \cite{Dvu2}, $n$-perfect pseudo MV-algebras in \cite{DDT, Dvu3}, perfect MV-algebras in \cite{DiLe}, and $n$-perfect pseudo effect algebras were introduced in \cite{DXY}. The aim of this section is to apply the relationships of RDP$_1$ of $\mathbb Z\lex G$ with RDP$_1$ of $G$ and with RDP$_1$ of the interval pseudo effect algebra $\Gamma(\mathbb Z\lex G,(n,0)),$ respectively, to simplify some results on strong $n$-perfect pseudo effect algebras from \cite{DXY}.

Let $n \ge 1$ be an integer. According to \cite{DXY}, a pseudo effect algebra $E$ is said to be  {\it $n$-perfect}  if
\begin{itemize}
\item[(i)] there exists an $n$-decomposition   of $E$ i.e., a system $(E_{0}, E_{1},\ldots, E_{n}),$ of subsets of $E$ such that  $E_i \cap E_j = \emptyset$ for $i\le j,$ $i,j \in \{0,1,\ldots,n\}$ and $\bigcup_{i=0}^n E_i = E.$

\item[(ii)] $E_{i}+ E_{j}$ exists   if $i+j<n.$

\item[(iii)] $E_{0}$ is a unique maximal ideal of $E.$
\end{itemize}

We note that a $1$-perfect pseudo effect algebra is said to be simply {\it perfect.}

Let $n>0$ be an integer. An element $a$ of a pseudo effect algebra $E$ is said to be {\it cyclic of order} $n>0$ if $na$ exists in $E$ and $na =
1.$ If $a$ is a cyclic element of order $n$, then $a^- = a^\sim$, indeed, $a^- = (n-1) a = a^\sim$.

We note that a group $G$ enjoys {\it unique  extraction of
roots} if, for all positive integers $n$ and $g,h \in G$, $g^n =
f^n$ implies $g=h$.  We recall that every linearly ordered group, or
a representable $\ell$-group, in particular every Abelian
$\ell$-group enjoys unique  extraction of roots, see \cite[Lem.
2.1.4]{Gla}.

Finally, according to \cite{DXY}, we define the following notion. Let $E$ be an $n$-perfect PEA satisfying (RDP)$_1$. We say that an $n$-perfect pseudo effect algebra $E$ with RDP$_1$ is a {\it
strong $n$-perfect} pseudo effect algebra   if

\begin{itemize}
\item[(i)] there exists a torsion-free unital po-group $(H,u)$ such that
$E=\Gamma(H,u),$

\item[(ii)] there exists an element $c\in E_{1}$ such that (a) $nc=u,$ and (b) $c\in C(H),$ where $C(H)=\{h \in H: h+x=x+h\ \forall x \in H\}.$
\end{itemize}

\begin{theorem}\label{th:4.1}
Let $E$ be a pseudo effect algebra  and $n \ge 1$ be an integer. Then $E$ is a strong $n$-perfect pseudo effect algebra if and only if there exists a torsion-free directed po-group $G$ with \RDP$_1$ such that $E\cong \Gamma(\mathbb Z\lex G,(n,0)).$ In addition, $G$ is unique.
\end{theorem}

\begin{proof} According to Proposition \ref{pr:2.3} and Theorem \ref{th:2.8}, we see that the conditions of \cite[Thm 6.7]{DXY} are satisfied. Therefore, calling that result, we have proved the statement in question.
\end{proof}

This theorem was proved originally in \cite[Thm 6.7]{DXY} under more general conditions, namely it was necessary to suppose that $G$ satisfies \RDP$_1$ as well as $\mathbb Z\lex G$ satisfies RDP$_1$. According to Proposition \ref{pr:2.3} and Theorem \ref{th:2.8}, we see that the second  condition is in \cite[Thm 6.7]{DXY} superfluous.

Analogously, we can redefine the categorical equivalence from \cite[Thm 6.10]{DXY} as follows.  Let $n\ge 1$ be a fixed integer. Let $\mathcal{G}$ be the category whose objects are  torsion-free directed  po-groups $G$ with \RDP$_1$   and morphisms are  po-group homomorphisms. Let $\mathcal{SPPEA}_{n}$ be the category whose objects are strong
$n$-perfect pseudo effect algebras and morphisms are homomorphisms of pseudo effect algebras.

We define a functor
$\mathcal{E}_n:\mathcal{G}\rightarrow\mathcal{SPPEA}_{n}$ as follows: for $G\in \mathcal{G},$ let
$\mathcal{E}_n(G):=\Gamma(\mathbb{Z}\overrightarrow{\times}G,(n,0))$
and if $h$ is a group homomorphism with domain $G,$ we set
$$\mathcal{E}_n(h)(x)=(i,h((ic)/x)),$$
where $c$ is a unique strong cyclic
element of order $n$ in $E.$

Then the functor
$\mathcal{E}_n$ is a faithful and full functor from the category $\mathcal{G}$ of directed torsion-free po-groups with RDP$_1$ into the category
$\mathcal{SPPEA}_{n}$ of strong $n$-perfect pseudo effect algebras. In addition, we have the following result which improves \cite[Thm 6.7]{DXY}.

\begin{theorem}\label{th:4.2}
The functor $\mathcal{E}_n$ defines a categorical equivalence of the category $\mathcal{G}$ of directed torsion-free  po-groups with  {\rm RDP}$_1$ and the category
$\mathcal{SPPEA}_{n}$  of strong $n$-perfect pseudo-effect algebras.
\end{theorem}

\begin{proof}
It follows directly from Theorem \ref{th:4.1} and \cite[Thm 6.7]{DXY}.
\end{proof}

\end{document}